\documentclass[11pt, letterpaper]{article}

\usepackage{graphicx}
\usepackage{amsmath}
\usepackage{amssymb}
\usepackage{setspace}
\usepackage{epstopdf}
\usepackage{color}
\usepackage[letterpaper,left=1in,right=1in,top=1in,bottom=1in]{geometry}
\usepackage[linesnumbered,ruled,vlined,longend]{algorithm2e}
\usepackage{multirow}
\usepackage{longtable}
\usepackage{rotating}
\usepackage{threeparttable}
\usepackage{colortbl}
\usepackage{amsthm}
\usepackage[modulo]{lineno}
\usepackage{enumerate}

\newtheorem{remark}{Remark}

\newtheorem{proposition}{Proposition}

\newtheorem{theorem}{Theorem}

\newtheorem{definition}{Definition}
\newtheorem{assumption}{Assumption}

\newcommand{\m}{\mathbb}

\title{\Large \bf Robust economic model predictive control with zone tracking}

\author{\centerline{\normalsize Benjamin Decardi-Nelson$^{a}$, Jinfeng Liu$^{a,}$\thanks{Corresponding author: J. Liu. Tel: +1-780-492-1317. Fax: +1-780-492-2881. Email: jinfeng@ualberta.ca}}\vspace{5mm}\\
    \centerline{\small $^{a}$ Department of Chemical \& Materials Engineering, University of Alberta,}\\
    \centerline{\small Edmonton, AB, Canada, T6G 1H9}}

\allowdisplaybreaks

\begin{document}

\date{}

\maketitle
\setstretch{1.39}

\begin{abstract}
This paper presents a robust economic model predictive control (EMPC) formulation with zone tracking for discrete-time uncertain nonlinear systems. The proposed design ensures that the zone tracking objective is achieved in finite steps and at the same time optimizes the economic performance. In the proposed design, instead of tracking the original target zone, a robust control invariant set within the target zone is determined and is used as the actual zone tracked in the proposed EMPC. This approach ensures that the zone tracking objective is achieved within finite steps and once the zone tracking objective is achieved (the system state enters the robust control invariant set), the system state does not come out of the target zone anymore. To optimize the economic performance within the zone in the presence of disturbances, we introduce the notion of risk factor in the controller design. An algorithm to determine the economic zone to be tracked is provided. The risk factor determines the conservativeness of the controller and provides a way to tune the EMPC for better economic performance. A nonlinear chemical example is presented to demonstrate the performance of the proposed formulation.
\end{abstract}

\noindent{\bf Keywords:} Predictive control; robustness; zone tracking; robust control invariant; nonlinear systems. 


\section{Introduction} \label{sec:introduction}
Nonlinear model predictive control (MPC) with a general objective known as economic MPC (EMPC) has received significant attention in recent years \cite{rawlings2012,liu2016,ellis2014}. The objective function in an EMPC generally reflects some economic performance criterion such as profit maximization or heat minimization. This is in contrast with the tracking MPC where the objective is a positive definite quadratic function. The integration of process economics directly in the control layer makes EMPC of interest in many areas especially in the process industry. There has been a significant number of applications of EMPC \cite{liu2015,decardi2018,zhang2020,griffith2017}. To address stability and computational issues of EMPC, different formulations of EMPC has been proposed \cite{angeli2011,liu2016,ellis2014,amrit2011}. 

Uncertainties arise as a result of imperfect models and/or unmeasured disturbances. The presence of uncertainties in any control system can result in performance degradation and/or loss of feasibility which can lead to loss of stability. Due to the integration of process economics in the control layer, it is not fully understood how the presence of uncertainties affect the economic performance of EMPC. In the context of tracking MPC, robust MPC is a common approach used to address the robustness of a control system in the presence of uncertainties. See \cite{mayne2016} for a recent survey on robust MPC as well as the associated challenges. Robust MPC techniques have also been applied to EMPC in the literature. In \cite{huang2012}, an EMPC formulation which is based on robust tracking of a prior nominal trajectory was proposed. In \cite{lucia2014}, a robust EMPC formulation based on scenario tree approach was presented. In \cite{mao2020}, a min-max robust EMPC algorithm was proposed to address transmission delays in networked control systems. Tube-based formulations with and without stochastic information have also been proposed \cite{bayer2014,bayer2016,dong2018}. However, they either use a min-max optimization approach or use the nominal model with tightened invariant constraints. In both cases, the computational demands are very high even for linear systems. 

While robust MPC techniques are common in the design of tracking MPCs for handling uncertainty, it was pointed out that simply transferring robust MPC techniques to EMPC could result in poor economic performance \cite{bayer2014}. This is because economic optimization and robustness are two objectives and often may conflict with each other. Robust MPC techniques have been designed to reject all disturbances to achieve their desired goal which may not be the case in EMPC as some disturbances can lead to better economic performance. In our previous work, we proposed an EMPC with zone tracking scheme to handle the two objectives in one integrated framework \cite{liu2018}. The use of a target zone allows for flexible handling of multiple objectives in the controller design and at the same time improves the degree of robustness of the controller due to the inherent robustness of zone control. It is worth noting that the concept of zone control is not new. Zone MPC have been reported in several areas such as diabetes treatment \cite{grosman2010}, control of building heating system \cite{privara2011}, control of irrigation systems \cite{mao2018} and coal-fired boiler-turbine generating system \cite{zhang2020}.  In the context of MPC literature, zone control is often dismissed as a trick to avoid feasibility issues and has received less attention in terms of theoretical analysis. A recent study on the stability analysis of MPC with generalized zone tracking \cite{liu2019} paves the way for further development of zone control.

In \cite{liu2018}, the stability and economic performance of the EMPC with zone tracking framework were studied without considering process uncertainty. In this work, we extend \cite{liu2018} to consider constrained nonlinear systems subject to unmeasured but bounded disturbances.  Instead of tracking the original target zone, we propose to track a robust control invariant set within the target zone so that once the system state enters the invariant zone, it will not exit the target zone anymore. The proposed design can ensure that the zone tracking objective is achieved in finite steps and at the same time optimizes the economic performance. It is found that in the presence of uncertainty, the economic performance of EMPC not only depends on the optimal steady state but also the size of the tracked zone. To take this into account, we introduce the notion of risk factor in the controller design. The risk factor determines the conservativeness of the controller and provides a way to tune the EMPC for better economic performance. An algorithm to determine the zone integrating the risk factor is also proposed. A nonlinear chemical example is presented to demonstrate the performance of the proposed formulation. We note that the extension of \cite{liu2018} to this work is not a trivial task. The contributions of this work include:
\begin{itemize}
    \item A detailed EMPC formulation with zone tracking for nonlinear systems subject to bounded disturbances
    \item The introduction of risk factor into the design to achieve improved economic performance in the presence of disturbances
    \item A description of an algorithm to determine the economic zone to be tracked
    \item A rigorous analysis of the feasibility, convergence and stability of the proposed design 
    \item A simulation example with extensive simulations illustrating the effectiveness of the proposed design
\end{itemize}

The remainder of this paper is organized as follows. Section 2 presents the preliminaries and describes the problem to be tackled. Section 3 describes the proposed robust economic MPC framework and the computation of the economic zone. Stability analysis for the proposed formulation is carried in Section 4. In Section 5, a numerical example is used to demonstrate the performance of the proposed approach. Section 6 summarizes the findings of this work and presents the future outlook.

\section{Preliminaries} \label{sec:preliminaries}

\subsection{Notation}
Throughout this work, the symbol $\m{I}_{\geq 0}$ denotes the set of nonnegative integers $\{0,1,2,\hdots\}$. $\m{I}_M^N$ is the set of integers from $M$ to $N$: $\m{I}_M^N = \{M,M+1,\hdots\,N\}$. $|\cdot|$ denotes the Euclidean norm of a scalar or a vector. $\|\cdot\|_n$ denotes the $n$-norm of a scalar or vector. A continuous function $\alpha:[0,a) \rightarrow [0,\infty)$ is said to belong to class $\mathcal{K}$ if it is strictly increasing and statisfies $\alpha(0)=0$. A class $\mathcal{K}$ function $\alpha$ is called a class $\mathcal{K}_\infty$ function if $\alpha$ is unbounded. A continuous function $\sigma:[0,\infty) \rightarrow [0,a)$ is said to belong to class $\mathcal{L}$ if it is strictly decreasing and statisfies lim$_{x \rightarrow \infty}\sigma(x)=0$. A continuous function $\beta:[0,a) \times [0,\infty)\rightarrow [0,\infty)$ is said to belong to class $\mathcal{KL}$ if for each fixed $r$, $\beta(r,s)$ is a class $\mathcal{L}$, and for each fixed $s$, $\beta(r,s)$ is a $\mathcal{K}$ function. The operator `$/$' means set subtraction such that $\m A/\m B = \{x\in R^{n_x}: x\in \m A, x \notin \m B\}$. 

\subsection{System description and control problem formulation}

In this work, we consider discrete-time nonlinear systems described by the following state-space model:
\begin{equation}\label{eqn:actual_system}
    x(n+1) = f(x(n),u(n),w(n)) 
\end{equation} 
where $x(n) \in \m{R}^{n_x}$ is the system state vector at time instant $n \in \m{I}_{\geq 0}$, $u(n) \in \m{R}^{n_u}$ is the control input vector and $w(n) \in \m{R}^{n_w}$ denotes the system disturbance vector.  It is assumed that the system state and the control input vectors are restricted to be in the coupled non-empty convex set of the following form:
 \begin{equation} \label{eqn:constraints}
    (x(n),u(n)) \in \m{Z} \subseteq \m{X} \times \m{U}
 \end{equation}
It is also assumed that the disburance is unknown and contained in a set $\m{W}$ ($w\in \m{W}$) where 
\begin{equation*}
    \mathbb{W} := \{ w \in \m{R}^{n_w}: \| w \|_{\infty} \leq \theta, \theta > 0 \}
\end{equation*}
with $\theta$ being a positive real number. Throughout this paper, we make the following assumptions.

 \begin{assumption}[Compact constraints] \label{asp:compact_constraints}
     The sets $\m{X}$, $\m{U}$ and $\m{W}$ are compact with $\m{W}$ containing the origin in its interior.
 \end{assumption}

 \begin{assumption}[Continuity]\label{asp:continuity_model}
     The function $f:\m{R}^{n_x} \times \m{R}^{n_u} \times \m{R}^{n_w} \rightarrow \m{R}^{n_x}$ is locally Lipschitz with respect to $x$ and $w$ for all $x\in \m Z$, $u\in \m U$, $w\in \m W$. This implies that there exist positive constants $L_x$ and $L_w$ such that:
     \begin{equation} \label{eqn:lipschiz_continuity}
        | f(x, u, w) - f(z,u,0) | \leq L_w |w| + L_x | x - z |   
     \end{equation} 
for all $x,z \in \m{X}$, $u \in \m{U}$ and $w \in \m{W}$.
 \end{assumption}

 We also recall the following definitions on forward invariant sets and robust control invariant sets. These concepts will be used in the description of the proposed EMPC design.

 \begin{definition}[Foward invariant set \cite{blanchini1999}] \label{postively_invariant_set}
     A set $\m X_r\subseteq \m{X}$ is said to be a forward or positively invariant set of the system $x(n+1)=f(x(n))$ if for every $x(n) \in \m X_r$, $f(x(n)) \in \m X_r$. 
 \end{definition}

 \begin{definition}[Robust control invariant set \cite{blanchini1999}]\label{robust_control_invariant_set}
     A set $\m X_r \subseteq \m{X}$ is said to be a robust control invariant set (RCIS) for system (\ref{eqn:actual_system}) if for every $x(n) \in \m X_r$, there exists a feedback control law $u(n) = \mu(x(n)) \in \m{U}$ such that $\m X_r$ is forward invariant for all $w(n) \in \m{W}$. 
 \end{definition}

The primary control objective of this work is to design a feedback controller such that it can drive the state of system~(\ref{eqn:actual_system}) to a pre-determined target zone $\m X_t\subset \m X$ if the initial state of system~(\ref{eqn:actual_system}) is outside the target zone ($x(0)\in\m X/\m X_t$) and maintain the state of system~(\ref{eqn:actual_system}) within the target set $\m X_t$ when the zone tracking is achieved. A secondary objective is to minimize the average economic cost over the infinite horizon $T$ characertized as follows:
 \begin{equation} \label{average_cost0}
     \limsup_{T \rightarrow \infty} \frac{1}{T} \sum_{n=0}^{T-1} \ell_e(x(n),u(n))
 \end{equation}
 where $\ell_e:\mathbb{R}^{n_x} \times \mathbb{R}^{n_u} \rightarrow \mathbb{R}$ is a general economic stage cost which is not necessarily quadratic or positive definite. In order to achieve the above control objectives, we resort to EMPC with zone tracking \cite{liu2018} and takes into account the presence of process disturbance $w$ in the design of the EMPC.

\section{Robust EMPC with zone tracking}

In this section, we present the design of the proposed robust EMPC with zone tracking scheme. The proposed design steers the system state to the target zone and optimizes the economic objective in the process. 

Given that there is model uncertainty due to the presence of process disturbances and that the target zone $\m X_t$ is not necessarily control invariant, a robust control invaraint set $\m X_e$ within the target zone ($\m X_e\subseteq \m X_t$) is determined and is used as the actual tracking zone in the proposed EMPC design \cite{liu2019}. To optimize the economic objective within the robust control invariant set, the set $\m X_e$ is optimized also according to the economic objective. In the remainder of this work, we will refer to this robust control invariant set $\m X_e$ as the economic zone. 


Let us first assume that such an economic zone $\m X_e$ has been determined. The procedure to create such an economic zone will be discussed in section~\ref{sec:economiczone}. It is also assumed that there is a steady state $(x_s,u_s)$ with $x_s\in\m X_e$, $u_s\in\m U$ such that it solves the following steady-state optimization problem: 
\begin{subequations} \label{eqn:ss_opt}
    \begin{align}
        (x_s,u_s) &= \arg \min ~ \ell_e(x,u) \label{eqn:ss_opt_cost}\\
        s.t. ~~~ & x = f(x,u,0) \label{eqn:ss_opt_model}\\ 
        & (x,u) \in \m{X}_e \times \m{U} \label{eqn:ss_constraint}
    \end{align}
\end{subequations}
Without loss of generality, we assume that $(x_s,u_s)$ is the unique solution to the above steady-state optimization problem.

\subsection{Design of the proposed EMPC with zone tracking}

With information about the current state $x(n)$, the proposed EMPC uses the nominal model of system~(\ref{eqn:actual_system}):
\begin{equation} \label{eqn:nominal_system}
    z(k+1) = f(z(k),v(k),0)
\end{equation}  
with the initial condition $z(0) = x(n)$ to find a control sequence $ \textbf{v}=\{ v(0), \hdots, v(N-1)\}$ and the associated state sequence $\textbf{z}=\{ z(0), \hdots, z(N) \}$ over the entire prediction horizon $N$ to minimize the cost function:
\begin{equation} \label{eqn:cost_function}
    V_N(x(n),\textbf{v}) = \sum_{k=0}^{N-1} \ell(z(k),v(k)) 
\end{equation}
In (\ref{eqn:nominal_system}) and (\ref{eqn:cost_function}), $z(n) \in \m{X} \subseteq \m{R}^{n_x}$ and $v(n) \in \m{U} \subseteq \m{R}^{n_u}$ are the nominal state vector and computed control input vector respectively in the proposed EMPC. The stage cost $\ell(\cdot,\cdot)$ is defined as follows: 
\begin{equation} \label{eqn:stage_cost}
    \ell(z,v) = \ell_e(z,v) + \ell_z(z)
\end{equation}
where $\ell_e(\cdot,\cdot)$ is the economic stage cost as introduced in (\ref{average_cost0}) and $\ell_z(\cdot)$ is a zone tracking penalty term which is defined as below:
\begin{subequations} \label{eqn:zone_penalty_opt}
    \begin{align} 
        \ell_z(z) = \min_{z^z} & ~~~ c_1(\| z - z^z \|_1) + c_2(\| z - z^z \|_2^2) \label{eqn:zone_penalty_opt_a}\\
        s.t. & ~~~ z^z \in \m{X}_e \label{eqn:zone_penalty_opt_b}
    \end{align}
\end{subequations}
with $c_1 \in \m{R}_{\geq 0}$, $c_2 \in \m{R}_{\geq 0}$ being non-negative weights on the $l_1$ norm and the squared $l_2$ norm respectively, $z^z$ is a slack variable and $\m{X}_e$ is the economic zone to be tracked. The zone tracking stage cost reflects the distance of the system states from the economic zone and is positive definite.

At each sampling time, the following dynamic optimization problem $\mathcal{P}_N(x(n))$ is solved:
\begin{subequations} \label{eqn:zone_opt}
    \begin{align} 
        \min_{\bf{v}} & ~~~ V_N(x(n),\textbf{v}) \label{eqn:zone_opt_a}\\
        s.t. & ~~~ z(k+1) = f(z(k),v(k),0), ~~~ k = 0, \hdots, N-1 \label{eqn:zone_opt_b}\\
        & ~~~ z(0) = x(n) \label{eqn:zone_opt_c}\\
        & ~~~  z(k) \in \m{X}, ~~~ k = 0, \hdots, N-1 \label{eqn:zone_opt_d}\\
        & ~~~ v(k) \in \m{U}, ~~~ k = 0, \hdots, N-1  \label{eqn:zone_opt_e}\\
        & ~~~ z(N) = x_s \label{eqn:zone_opt_f}
    \end{align}
\end{subequations}
In the optimization problem (\ref{eqn:zone_opt}) above, Equation (\ref{eqn:zone_opt_c}) is the initial state constraint, Equation (\ref{eqn:zone_opt_f}) is a terminal equality constraint and Equations (\ref{eqn:zone_opt_d}) and (\ref{eqn:zone_opt_e}) are the constraints on the state and inputs respectively. As a result of the cost function employed, the optimization problem in (\ref{eqn:zone_opt}) is a multi-objective optimization problem which seeks to minimize the deviation of the system's state from the economic zone $\m{X}_e$ while optimizing the economic objective.



The solution of $\mathcal{P}_N(x(n))$ denoted $\mathbf{v^*}$ gives  an optimal value of the cost $V_N^0 (x(k))$ and at the same time $u(n)=v^*(0)$ is applied to the actual system (\ref{eqn:actual_system}). Notice that the nominal system is used in the optimization and therefore generates mismatch between the prediction in the EMPC optimization and the actual system evolution. We will show in the next section that under some mild conditions, our proposed controller is able to stabilize the plant in the presence of this mismatch. The prediction horizon is shifted forward by one sampling time once information about $x(n+1)$ is known and the optimization problem $\mathcal{P}_N(x(n+1))$ is solved to find $u(n+1)$. 

\subsection{Construction of the economic zone}\label{sec:economiczone}

In the previous section, we have presented the proposed EMPC formulation with zone tracking. In the proposed design, a robust control invariant economic zone $\m X_e$ replaces the original target zone and is the zone to be tracked. In this section, we discuss how to construct the economic zone. 

\subsubsection{Risk factor}

In determining the economic zone $\m X_e$, the idea is to find a robust control invariant set within the original target zone $\m X_t$ while taking into account the economic performance of the system within the control invariant set. The use of a robust control invariant set as the actual tracking zone ensures that the system state converges to the zone and will not leave the zone again once enters the invariant zone even in the presence of disturbances. This will be shown in the stability analysis section. 

While any robust control invariant set within the original target zone can achieve the zone tracking objective, the size of the robust control invariant set affects the economic performance of the system. Due to the presence of disturbance in the system, the overall economic perofrmance of the system not only depends on the optimal steady state within the zone $\m X_e$ but also depends on the economic performance of the system within the zone. In order to account for this in determining the economic zone, we introduce the concept of risk factor $\delta \in \m{R}$ in the $\m X_e$ construction. The risk factor is a positive scalar that can be tuned. It determines the size of the economic zone and ultimately, the conservativeness of the controller. When a higher risk factor is used, the size of the economic zone is larger and the controller is less conservative. Algorithm~\ref{alg:ezone} presented in the next subsection will summarize how the risk factor is used in the computation of the economic zone. 

\subsubsection{Computing the economic zone}

The algorithm for determining the economic zone builds on the graph-based robust control invariant set computing algorithm developed in \cite{decardi2021}. In the algorithm, the state space $\m X$ is quantized into small closed sets $B_i,i=1,\hdots l$. The collection of these cells, $\mathcal{C} = \{ B_1, \hdots, B_l \}$, is called the finite covering of the state space $\m{X}$. The closed sets in the finite covering $\mathcal{C}$ are also known as cells or boxes such that:
\begin{subequations}
    \begin{align}
         & \m{X} \subseteq \cup_{B_i \in \mathcal{C}} B_i  \\ 
         & B_i \cap B_j = \emptyset, ~ \forall B_i,B_j \in \mathcal{C} ~ \text{with} ~ i \neq j 
    \end{align}
\end{subequations}
Following the quantization, the system dynamics is approximated using a directed graph $G$. Graph investigations are then carried out on the directed graph to determine the cells that approximate the largest robust control invariant set while the ones that do not form part of the robust control invariant set are discarded. 

We denote by $\mathcal{C}_r$ the cells that approximate the robust control invariant set.

\begin{algorithm} \label{alg:ezone}
    \caption{Determination of economic zone}
    \KwIn{ $f$, $\m X_t$, $\m{U}$, $\m{W}$, $\ell_e$, $\delta$}
    \KwOut{ $\m{X}_e$ }

    Create a finite convering $\mathcal{C}_t$ of $\m X_t$ \\

    Initialize the cells that satisfy the economic criterion $\mathcal{C}_e$ as empty array \\

    \For{$B_i$ in $\mathcal{C}_t$}{
        \If{$ \forall x \in B_i,  \exists u \in \m{U}: \ell_e(x+f(x,u,w)-f(x,u,0),u) \leq \delta, \; \forall w \in \m{W}$}{
            Add $B_i$ to $\mathcal{C}_e$
        }
    }
    Initialize Algorithm 2 in \cite{decardi2021} with $\mathcal{C}_e$ \\
    Compute an inner approximation of the largest robust control invariant set $\mathcal{C}_r$ contained in $\mathcal{C}_e$\\

    $\m{X}_e \leftarrow \cup_{B_i \in \mathcal{C}_r} B_i$ \\
    
    \textbf{return} $\m{X}_e$
\end{algorithm}

The procedure for determining the economic zone $\m{X}_e$ is summarized in Algorithm \ref{alg:ezone}. The algorithm has a few inputs including the system model $f$, the risk factor $\delta$, the economic objective $\ell_e$ as well as the input and the disturbance sets. The algorithm returns the calculated economic zone $\m{X}_e$. 

Intuitively, the algorithm seeks to find an economic zone that compensates for the effects of the disturbances on the economics of the closed-loop system while ensuring good stability property. This is achieved by backing-off from the boundaries of the target zone to obtain $\m X_e$. The algorithm is in two main steps. In the first step of the algorithm, the target zone is quantized with the help of a finite covering $\mathcal{C}_t$ and then the set of cells $\mathcal{C}_e$ within the target zone $\m X_t$ that satisfy the economic criterion is determined. Consider a cell $B_i \in \mathcal{C}_t$, if 
\begin{equation}
    \forall x \in B_i,  \exists u \in \m{U}: \ell_e(x+f(x,u,w)-f(x,u,0),u) \leq \delta, \; \forall w \in \m{W}
\end{equation}
then the cell $B_i$ is added to $\mathcal{C}_e$. The remainder of the cells in $\mathcal{C}_t$ are then discarded. The choice of the selection criterion in Algorithm \ref{alg:ezone} stems from the fact that every state within the target zone is a potential initial state as well as a potential end state after one time-step. We focus on the latter since our proposed controller does not consider the effects of the disturbance. The idea is that, for any potential end state given by the nominal system, we know that the disturbance will be applied in the real system. Thus, we are taking into consideration the effects of the disturbance on the economics implicitly. By considering the end state in the selection criterion, we want to guarantee that the economic performance of the closed-loop system is bounded above by the risk factor $\delta$ irrespective of the disturbance $w$. 

It is worth mentioning that the set formed by the union of the cells in $\mathcal{C}_e$ is not necessarily robust control invariant. The second step, therefore seeks to address this by finding the cells in $\mathcal{C}_e$ that inner approximate the largest robust control invariant set. This is achieved by initializing Algorithm 2 in \cite{decardi2021} with $\mathcal{C}_e$ and then using the algorithm to find the cells that inner approximate the largest robust control invariant set $\mathcal{C}_r$ contained in the set formed by $\mathcal{C}_e$.

\begin{remark}
    
The cells in $\mathcal{C}_r$ needs to be combined and represented in a way that makes the optimization problem presented in (\ref{eqn:zone_opt}) easier to solve. One of such representations is to find an inner approximation convex hull of the cells in $\mathcal{C}_r$ if the cells form a convex set. Another approach is to use a more general set representation such as alpha hull. However, this may lead to the use of non-convex sets in (\ref{eqn:zone_opt}) which can increase the complexity of the optimization problem. 

\end{remark}

\begin{remark} 
As a result of the presence of the disturbances, it is in general difficult to determine the optimal control inputs. One approach to determine the optimal feedback control law is to solve a min-max optimization problem \cite{mayne2016}. However, it suffers from high computational demand which makes it challenging to implement. In this paper, we propose an EMPC scheme based on only the nominal model and zone tracking. The zone to be tracked can be considered as an economic trust region. This makes our propsed approach similar to other trust-region based approaches such as the Lyapunov-based EMPC \cite{heidarinejad2012} and that presented in \cite{griffith2017}. However, in the proposed formulation, we do not make use of any additional constraints such as Lyapunov constraints in the formulation. Moreover, our formulation introduces economic risk factor in the controller design thus implicitly considers an upper bound on the asymptotic average performance of the closed-loop system.
\end{remark}

\section{Stability analysis}

In this section, we address the stability of the proposed control algorithm. To proceed with the discussion, we first introduce a few relevant definitions and assumptions.  

First, we define the $N$-step reachable set of the optimal steady-state $x_s$ based on the nominal model. The $N$-step reachable set will be used to construct a set for the initial state of the system to esnure the feasibility of the proposed EMPC. 

\begin{definition}[$N$-step reachable set] Consider the nominal system of system~(\ref{eqn:actual_system}) (i.e., $w\equiv 0$ for all time). A set $\m X_N$ is called the $N$-step reachable set with respect to the steady-state $x_s$ if it contains all the states that can be steered to $x_s$ in $N$ steps while satisfying the state and input constraints. That is, 
    \begin{equation}\label{eqn:xn}
        \m X_N=\{x(0)\in \m X|\exists(x(n), u(n))\in \m Z, n\in \m I_0^{N-1}, \textmd{such that } x(N)=x_s\}
    \end{equation}
\end{definition}

\begin{assumption}\label{ass:xncompactness}
The $N$-step reachable set $\m X_N$ is a compact set with $x_s$ in the interior of set. 
\end{assumption}

Next, we introduce the definition of dissipative systems and the relevant assumptions. These definition and assumptions will be used to establish the stability of the proposed EMPC.

\begin{definition}[Strictly dissipative systems] 
    The nominal system $\tilde x(n+1) = f(\tilde x(n),u(n),0)$ is strictly disspipative with respect to the supply rate $s:\mathbb{X} \times \mathbb{U} \rightarrow \mathbb{R}$ if there exists a continuous storage function $\lambda(\cdot):\mathbb{X} \rightarrow \mathbb{R}$ and a $\mathcal{K}_{\infty}$ function $\alpha(\cdot)$ such that the following hold for all $\tilde x \in \mathbb{X}$ and $u \in \mathbb{U}$:
    \begin{equation} \label{eqn:strict_dissipativity}
        \lambda(f(\tilde x,u)) - \lambda(\tilde x) \leq s(\tilde x,u) - \alpha(|\tilde x - x_s|)
    \end{equation}
\end{definition}

\begin{assumption}[Strict disspativity] \label{asp:strict_dissipativity}
    The nominal system $\tilde x(n+1) = f(\tilde x(n),u(n),0)$ is strictly disspative with respect to the supply rate 
    \begin{equation*}
        s(\tilde x,u) = \ell_e(\tilde x,u) - \ell_e(x_s,u_s)
    \end{equation*}
\end{assumption}

\begin{assumption}[Weak controllability] \label{asp:weak_controllability}
    There exists a $\mathcal{K}_\infty$ function $\gamma(\cdot)$ such that for all $x \in \m{X}_N$, there exists a feasible solution to (\ref{eqn:zone_opt}) such that $\displaystyle\sum_{k=0}^{N-1} |v(k) - u_s| \leq \gamma(|x-x_s|)$.
\end{assumption}

The following proposition provides an upper bound on the deviation of the nominal system state trajectory from the uncertain system state trajectory when the same input sequence is applied. 

\begin{proposition} \label{prop:deviation_upper_bound}
    Consider the following system 
    \begin{equation}
        x(n+1) = f(x(n),u(n),w(n))
    \end{equation}
    and the corresponding nominal system
    \begin{equation}
        \tilde x(n+1) = f(\tilde x(n),u(n),0)
    \end{equation}
    with the initial condition $x(n)=\tilde x(n) \in \m{X}$. The deviation of the nominal system state $\tilde{x}$ from the state $x$ over one sampling time is bounded as follows:
    \begin{equation} \label{eqn:error_growth}
        |x(n+1) - \tilde x(n+1)| \leq \sqrt{n_x}L_w \theta 
    \end{equation} 
    for all $x(n), \tilde x(n) \in \m{X}$ and all $w(n) \in \m{W}$. 
\end{proposition}

\begin{proof}
Let us define the deviation of $\tilde x$ from $x$ as $e$ such that $e=x-\tilde x$. Therefore, $e(n+1) = x(n+1) - \tilde x(n+1)$, which can further be written as:
\begin{equation}
    e(n+1) = f(x(n),u(n),w(n)) - f(\tilde x(n),u(n),0)
\end{equation}
Taking the Euclidean norm of the error $e$ and applying (\ref{eqn:lipschiz_continuity}), the following inequality is obtained
\begin{equation} \label{eqn:error_system}
    |e(n+1)| \leq L_w |w(n)| + L_x |x(n) - \tilde x(n)| = L_w|w(n)| + L_x |e(n)|
\end{equation}
    for all $x(n), \tilde x(n) \in \m{X}$ and $w(n) \in \m{W}$. Since the initial state for both the nominal and the uncertain system are the same i.e. $x(n) = \tilde x(n)$, we have that the initial deviation is 0, i.e. $e(n)=0$. Given that $\|w\|_{\infty}\leq \theta$, $|w|\leq \sqrt{n_x}\theta$. This leads to (\ref{eqn:error_growth}) and proves Proposition \ref{prop:deviation_upper_bound}.
\end{proof}

We now state the main results of this section. Theorem~\ref{thm:1} considers the nominal system of system~(\ref{eqn:actual_system}) and finds a Lyapunov function of the system with respect to the steady state $x_s$. Theorem~\ref{thm:2} will use this Lyapunov function to study the uncertain system to establish the feasible region, finite step convergence, and ultimate stability and robustness of the proposed EMPC. 

\begin{theorem}\label{thm:1}
Consider the nominal system of system~(\ref{eqn:actual_system}) under the control of EMPC (\ref{eqn:zone_opt}). Suppose that Assumption~\ref{asp:strict_dissipativity} holds and $\lambda(\cdot)$, $\alpha(\cdot)$, $s(\tilde x,u)=\ell_e(\tilde x,u)-\ell_e(x_s,u_s)$ are the associated functions that satisfy the condition (\ref{eqn:strict_dissipativity}) for all $\tilde x \in \m{X}$ and $u \in \m{U}$. Define the rotated cost as follows:
\begin{equation} \label{eqn:rotated_cost}
\tilde{\ell}_e(\tilde x,u) = \ell_e(\tilde x,u) - \ell_e(x_s,u_s) + \lambda(\tilde{x}) - \lambda(f(\tilde x,u,0)).
\end{equation}
Then, the following dynamical optimization problem is equivalent to the proposed EMPC~(\ref{eqn:zone_opt}):
\begin{subequations} \label{eqn:new_zone_opt}
    \begin{align} 
        \min_{\bf{v}} & ~~~ \tilde{V}_N(\tilde x(n),\textbf{v}) = \sum_{k=0}^{N-1} \left(\tilde{\ell}_e(z(k),v(k)) + \ell_z(z(k))\right)\label{eqn:positive_definite_cost_function}\\
        s.t. & ~~~ (\ref{eqn:zone_opt_b}) - (\ref{eqn:zone_opt_f})
    \end{align}
\end{subequations}
If Assumption~\ref{asp:weak_controllability} also holds, then the value function of (\ref{eqn:new_zone_opt}) denoted as $\tilde V_N^0(\cdot)$ is a Lyapunov function of the closed-loop system under the control of EMPC~(\ref{eqn:zone_opt}) with respect to the optimal steady-state $x_s$
\end{theorem}

\begin{proof} In this proof, we use $\tilde x$ to denote the state of the nominal system under the control of the proposed EMPC. Based on the definition of the rotated cost as in (\ref{eqn:rotated_cost}) and the condition (\ref{eqn:strict_dissipativity}), it can be concluded that the rotated cost is bounded from below for all $(\tilde x,u) \in \m{Z}$ as follows:
    \begin{equation} \label{eqn:rotated_cost_bound}
        \tilde{\ell}_e(\tilde x,u) \geq \alpha(|\tilde x - x_s|)
    \end{equation}
    Based on the definition of $\ell(\cdot,\cdot)$ in (\ref{eqn:stage_cost}) and the rotated cost in (\ref{eqn:rotated_cost}), the stage cost $\ell(z,u)$ can be equivalently expressed as follows:
    \begin{equation} \label{eqn:auxilliary_stage_cost}
        {\ell}(z,v) =  \tilde{\ell}_e(z,v) + \ell_z(z) + \ell_e(x_s,u_s) - \lambda(z) + \lambda(f(z,v,0))
    \end{equation} 
Based on the above expression of $\ell(\cdot,\cdot)$, the cost function $V_N(\cdot)$ in the optimization problem (\ref{eqn:zone_opt}) at time $n$ can be equivalently expressed as follows:
\begin{equation} \label{eqn:rotated_cost_function}
        {V}_N(\tilde x(n),\textbf{v}) = \sum_{k=0}^{N-1} \left(\tilde{\ell}_e(z(k),v(k)) + \ell_z(z(k))\right) - \lambda(z(0)) + \lambda(z(N)) + N\ell_e(x_s,u_s)
    \end{equation}
    Taking into account the constraint~(\ref{eqn:zone_opt_f}) in EMPC~(\ref{eqn:zone_opt}), the last three terms in the above expression of $V_N(\cdot, \cdot)$ are constants. This implies that if we replace the cost function $V_N(\cdot,\cdot)$ in the EMPC optimization problem (\ref{eqn:zone_opt}) with the new cost function as in (\ref{eqn:positive_definite_cost_function}), the solution of the EMPC optimization problem remains the same. That is, the original EMPC~(\ref{eqn:zone_opt}) is equivalent to the new EMPC (\ref{eqn:new_zone_opt}). Let us denote the optimal value of the cost function (the value function) of the new EMPC~(\ref{eqn:new_zone_opt}) as $\tilde{V}^0_N(\tilde x(n))$. Taking into account (\ref{eqn:rotated_cost_bound}) and the expression of $\tilde{V}_N(\tilde x(n),\textbf{v})$, and noticing that $z(0) = \tilde x(n)$ in the EMPC optimization problem, it can be obtained that:
    \begin{equation}
        \tilde{V}^0_N(\tilde x(n)) \geq \tilde{\ell}_e(z(0),u({}0)) + \ell_z(z(0)) \geq \tilde{\ell}_e(z(0),u(0)) \geq \alpha(|\tilde x(n)-x_s|)
    \end{equation}
    From Assumption \ref{asp:weak_controllability}, there exists a $\beta(\cdot) \in \mathcal{K}_{\infty}$ such that for all $\tilde x(n) \in \m{X}_N$ (see Appendix of \cite{diehl2010}):
    \begin{equation} \label{bounded_above}
        \tilde{V}^0_N(\tilde x(n)) \leq \beta(|\tilde x(n)-x_s|)
    \end{equation}
    For the nominal system, it can be shown that the value function $\tilde{V}^0_N(\cdot)$ is non-increasing and satisfies the following condition:
    \begin{equation} \label{eqn:Lyapunov_nominal}
        \tilde{V}^0_N(\tilde x(n+1)) - \tilde{V}^0_N(\tilde x(n)) \leq - \tilde{\ell}_e(\tilde x(n),u(n)) - \ell_z(\tilde x(n),u(n)) \leq -\alpha(|\tilde x(n)-x_s|)
    \end{equation}
    This makes the value function $\tilde{V}^0_N(\cdot)$ a Lyapunov function of the closed-loop system under the control of EMPC~(\ref{eqn:zone_opt}) with respect to the optimal steady state $x_s$. This proves Theorem~\ref{thm:1}.
\end{proof}
Before presenting Theorem~\ref{thm:2}, we introduce the set $\Omega_\rho$ defined based on the level set of the Lyapunov function $\tilde V_N^0(\cdot)$:
\begin{equation}\label{eqn:levelset}
    \Omega_\rho = \{x\in \m X: \tilde V_N^0(x)\leq \rho\}.
\end{equation}
Based on the above definition, we also define $\Omega_{\rho_{\min}}$ as follows:
\begin{equation}\label{eqn:rhomin}
\Omega_{\rho_{\min}}:=\max\{\tilde V_N^0(x(n+1)): |x(n)-x_s|\leq \alpha^{-1}(K_VL_w\sqrt{n_x}\theta + Hn_xL_w^2\theta^2)\}
\end{equation}
where $K_V$ is a positive constant that bounds the the magnitude of the partial derivative $\dfrac{\partial \tilde V_N^0(x)}{\partial x}$ such that $|\dfrac{\partial \tilde V_N^0(x)}{\partial x}|\leq K_V$ for all $x\in \m X$, and $H$ is the constant associated with the Taylor expansion of $\tilde V_N^0(x)$ (which will be made clearer in the proof of Theorem~\ref{thm:2}). Further, we denote the maximum level set within $\m{X}_N$ as $\Omega_{{\rho}_{\max}}$. 

\begin{theorem} \label{thm:2}
    Consider system (\ref{eqn:actual_system}) in closed-loop with EMPC~(\ref{eqn:zone_opt}). Let the target zone and the economic zone satisfy: $\Omega_{\rho_{\min}} \subset \m X_e \subset \m X_t \subset \Omega_{\rho_{\max}} \subset \m X$. If Assumptions \ref{asp:continuity_model} -- \ref{asp:weak_controllability} hold, the magnitude of the partial derivative $\dfrac{\partial \tilde V_N^0(x)}{\partial x}$ is upper bounded such that $|\dfrac{\partial \tilde V_N^0(x)}{\partial x}|\leq K_V$ for all $x\in \m X$, and if there exist $\epsilon_s>0$, $\rho_s>0$ such that:
    \begin{equation}\label{eqn:condition}
        -\alpha(\rho_s) + K_VL_w\sqrt{n_x}\theta + Hn_xL_w^2\theta^2 \leq -\epsilon_s
    \end{equation}
    where $\alpha(\cdot)$ is a class $\mathcal{K}_{\infty}$ function associated with Assumption~\ref{asp:strict_dissipativity} and as defined in (\ref{eqn:strict_dissipativity}), and $H$ is the constant associated with the Taylor expansion of $\tilde V_N^0(x)$, then the closed-loop system state $x$ converges to the economic zone $\m X_e$ in finite steps and then maintains in $\m X_e$ all the time for any initial condition $x(0)\in \Omega_{\rho_{\max}}$. 
\end{theorem}

\begin{proof}
    In this proof, we consider applying EMPC~(\ref{eqn:zone_opt}) which is designed based on the nominal system to the actual system with disturbance $w$. At time instant $n$, the EMPC optimization problem is solved with the actual system state $x(n)$ as the initial condition and only the first input value in the optimal input trajectory is applied to the system. Applying Proposition~\ref{prop:deviation_upper_bound}, from $n$ to $n+1$, the deviation of the actual system state $x(n+1)$ from the nominal system state $\tilde x(n+1)$ is bounded as following:  
    \begin{equation}\label{eqn:prooftheta}
        |x(n+1)-\tilde x(n+1)| \leq  \sqrt{n_x}L_w \theta
    \end{equation}
Using Taylor expansion, we can obtain the following relation:
    \begin{equation}
        \tilde V_N^0(x(n+1)) = \tilde V_N^0(\tilde x(n+1)) + \left.\dfrac{\partial { \tilde V_N^0(x)}}{\partial x}\right|_{\tilde x(n+1)}(x(n+1) - \tilde x(n+1)) + H.O.T
    \end{equation}
where $H.O.T$ includes the high order terms in the above Taylor expansion. For $x\in \m X$, a positive constant $H$ can be found such that the high order terms satisfy the following constraint:
    \begin{equation}\label{eqn:proof_HOT}
        H.O.T\leq H|x(n+1) - \tilde x(n+1)|^2
    \end{equation}
    Taking into account that the initial condition ($\tilde x(n) = x(n)$) when solving the EMPC optimization, (\ref{eqn:Lyapunov_nominal}), (\ref{eqn:prooftheta})--(\ref{eqn:proof_HOT}), it can be derived the following inequality:
    \begin{equation}\label{eqn:proof_V}
        \tilde V_N^0(x(n+1)) - \tilde V_N^0(x(n)) \leq -\alpha(|x(n) - x_s|) + K_VL_w\sqrt{n_x}\theta + Hn_xL_w^2\theta^2
    \end{equation}
    If condition~(\ref{eqn:condition}) is satisfied, from (\ref{eqn:proof_V}), it can be seen that
    \begin{equation}\label{eqn:proof_V2}
        \tilde V_N^0(x(n+1)) - \tilde V_N^0(x(n)) \leq -\epsilon_s
    \end{equation}
    for all $x(n)\in \Omega_{\rho_{\max}}$ and $|x(n)-x_s|\geq \rho_s$. This implies that as long as $|x-x_s|\geq \rho_s$, the Lyapunov function keeps decreasing. By applying (\ref{eqn:proof_V2}) recursively, it is proved that the system state enters a region such that $|x-x_s|\geq \rho_s$ in finite steps. Given the definition of $\Omega_{\rho_{\min}}$, it futher implies that once the state satisfies $|x-x_s|\geq \rho_s$, the state will remain in $\Omega_{\rho_{\min}}$ all the time. Then, the actual system under the control of the proposed EMPC will eventually converge to $\Omega_\rho$. Given that $\Omega_{\rho_{\min}}\subset\m X_e$, this proves that the system state enters the economic zone in finite steps and then remains within the economic zone. This proves Theorem~\ref{thm:2}.
\end{proof}

\begin{remark}\label{rmk:boundary_steady_state}
The use of general economic objective in economic MPC may drive the system states to operate close to the operating constraints. This is no different in the zone economic MPC formulation. It is therefore possible that the optimal steady state within the desired economic zone is on the boundary. Since Theorem \ref{thm:2} require that the optimal operating point be in the interior of the desired economic zone, this needs to be resolved. One way to achieve this is to construct a smaller economic zone and then use that in the controller design. This way, the desired economic zone will be tracked once the smaller economic zone is tracked. Another approach is to construct an economic zone with a bigger risk factor and then track this new economic zone while ensuring that system's states go to the optimal operating point of the desired economic zone. To achieve this however, the cost function may need to be regularized to ensure that the steady-state point is tracked by the controller once the states are within the economic zone. 
\end{remark}


\section{Illustrative example}

In this section, we demonstrate the efficacy of our proposed controller using a chemical process. We first describe the chemical process example used in our analysis. Subsequently, we consider the impact of the risk factor on the asymptoptic average economic performance of our proposed controller and then finally compare the performance of our proposed controller to that of the conventional economic MPC.

\subsection{Process description}
Consider a well-mixed continuously stirred tank reactor (CSTR) where a single first-order irreversible reaction of the form $A \rightarrow B$ takes place. Since the reaction is exothermic, thermal energy is removed from the reactor through a cooling jacket. Assuming constant volume reaction mixture, the following nonlinear differential equations are obtained based on energy balance and component balance for reactant $A$:
\begin{subequations} \label{eqn:cstr_equations}
    \begin{align} 
        \frac{dC_A}{dt} ={} & \frac{q}{V}(C_{Af} - C_A) - k_0 \exp(-\frac{E}{RT})C_A \\
        \frac{dT}{dt} ={} & \frac{q}{V}(T_{f} - T) +  \frac{-\Delta H}{\rho C_p} k_0 \exp(-\frac{E}{RT})C_A + \frac{UA}{V \rho C_p}(T_c - T)
    \end{align}
\end{subequations}
where $C_A$ and $T$ denote the reactant concentration and temperature of the reaction mixture in $mol/L$ and $K$ respectively, $T_c$ denotes the temperature of the coolant stream in $K$, $q$ denotes the volumetric flow rate of the inlet and outlet streams of the reactor in $L/min$, $C_{Af}$ denotes the concentration of reactant $A$ in the feed stream, $V$ denotes the volume of the reaction mixture, $k_0$ denotes the reaction rate pre-exponential factor, $E$ denotes the activation energy, $R$ is the universal gas constant, $\rho$ is the density of the reaction mixture, $T_f$ is the temperature of the feed stream, $C_p$ is the specific heat capacity of the reaction mixture, $\Delta H$ is the heat of reaction and $UA$ is the heat transfer coefficient between the cooling jacket and the reactor. The values of the parameters used in the simulations are listed in Table \ref{tb:parameters}. A linear version of this model was used in \cite{bayer2016} in the context of robust tube-based economic MPC. 

\begin{table}[hbt!]
  \begin{center}
  \caption{Table of parameter values}\label{tb:parameters}
    \begin{tabular}{ccl}
    Parameter & Unit & Value \\\hline
    $q$ & $L/min$ & $100.0$ \\
    $V$ & $L$ & $100.0$ \\
    $c_{Af}$ & $mol/L$ & $1.0$ \\
    $T_f$ & $K$ & 350.0 \\
    $E/R$ & $K$ & $8750.0$ \\
    $k_0$ & $min^{-1}$ & $7.2 \times 10^{10}$ \\
    $-\Delta H$ & $J/mol$ & $5.0 \times 10^4$ \\
    $UA$ & $J/min\cdot K$ & $5.0 \times 10^4$ \\
    $c_p$ & $J/g\cdot K$ & $0.239$ \\
    $\rho$ & $g/L$ & $1000.0$ \\ \hline
    \end{tabular}
  \end{center}
\end{table}

The nonlinear model of (\ref{eqn:cstr_equations}) is discretized using a step-size  $h = 0.1$ $min$ to obtain a discrete-time nonlinear state space model of the following form:
\begin{equation}
    x(n+1) = f(x(n),u(n), w(n)) 
\end{equation}
where $x=[C_A ~ T]^T$ is the state vector, $u=T_c$ is the input and $w=[C_{Af} ~ T_f]^T$ is the disturbance vector. The state, input and disturbance are assumed to be subject to the following hard constraints: $0.0 \leq x_1 \leq 1.0$, $345.0 \leq x_2 \leq 355.0$, $285.0 \leq u \leq 315.0$, $0.9 \leq w_1 \leq 1.1$ and $348.0 \leq w_2 \leq 352.0$. The disturbances are assumed to be uniformly distributed in the constraints with their nominal values being $1.0$ and $350.0$ as shown in Table \ref{tb:parameters}.

The economic objective $\ell_e$ is to minimize the concentration of reactant $A$ (i.e. maximize the concentration of reactant $B$) in the reactor such that
\begin{equation}\label{eqn:economic_cost}
    \ell_e(x,u) = c_A
\end{equation}
To ensure that the economic cost is optimized within a reasonable temperature range, a zone tracking objective $\ell_z$ is incorporated into the control objective where
\begin{equation}
    \ell_z (x) = 
    \begin{cases}
        10.0\times(348.0 - T)^2 &~~~ \text{for } T < 348.0 \\
        0 & ~~~ \text{for } 348.0 \leq T \leq 352.0\\
        10.0\times(352.0 - T)^2 &~~~ \text{for } T > 352.0 
    \end{cases}
\end{equation}
 The zone tracking objective is a quadratic function that penalizes the deviation of the system states from the target zone $\m X_t$. The overall control objective therefore becomes 
\begin{equation}\label{eqn:control_objective}
    \ell(x,u) := \ell_e(x,u) + \ell_z(x) = c_A + 
    \begin{cases}
        10.0\times(348.0 - T)^2 &~~~ \text{for } T < 348.0 \\
        0 & ~~~ \text{for } 348.0 \leq T \leq 352.0\\
        10.0\times(352.0 - T)^2 &~~~ \text{for } T > 352.0 
    \end{cases}
\end{equation}
This control objective is multi-objective and can be achieved by manipulating the temperature of the coolant $T_c$ in the cooling jacket. Notice that in this example, only the temperature has a zone requirement. The zone on the concentration therefore spans the entire constraint. The safe operating region (hard constraints) as well as the target zone $\m X_t$ for this example is presented in Figure \ref{fig:original_zones}.

\begin{figure}[hbt!]
  \begin{center}
    \includegraphics[width=0.7\textwidth]{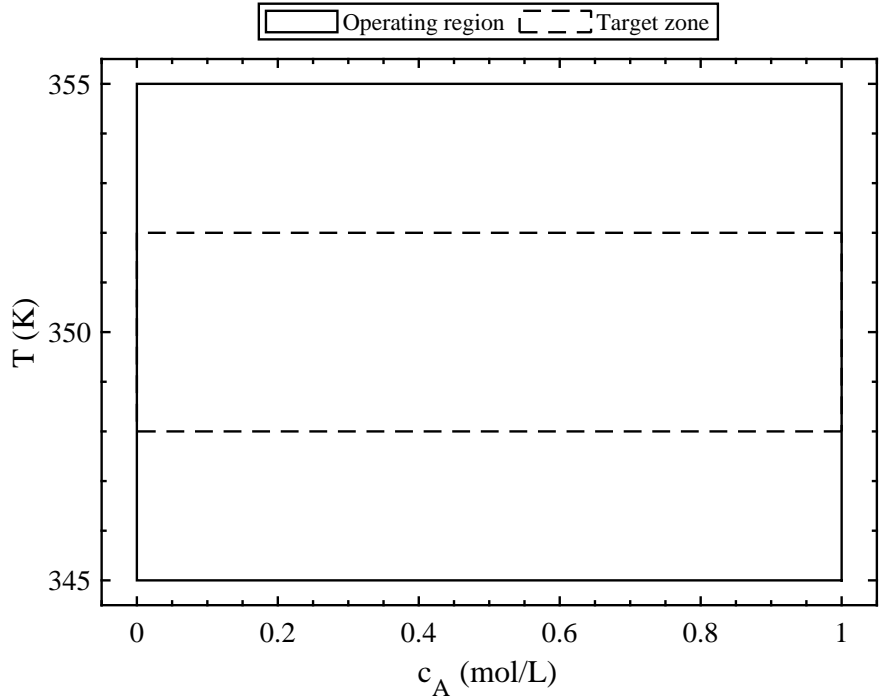}    
    \caption{The sets used in the controller design. The operating region (solid line) is the hard constraint on the states where the process must be operated within without any violation. The target zone (dashed line) is a soft constraint on the states which ensures that the economic cost is optimized within reasonable temperature bounds.} 
    \label{fig:original_zones}
  \end{center}
\end{figure}

In the simulations, unless otherwise stated, the control and prediction horizons of all controllers are $N=20$ respectively. The $l_1$ norm weight $c_1$ and $l_2$ norm weight $c_2$ in (\ref{eqn:zone_opt}) for our proposed controllers were chosen as 0 and 10 respectively. We assume that all the system states are available to the controller. The proposed robust economic MPC scheme and the traditional economic MPC scheme were numerically transcribed using the direct multiple shooting method and solved using IPOPT \cite{wachter2006}. The optimization problems were implemented in the modeling language JuMP \cite{DunningHuchetteLubin2017}. Each dynamic simulation was run for 1000 time steps and the asymptotic average performance, computed thereafter. 

In the subsequent analysis, the optimal steady-state economic cost $\ell_e^*$ was obtained by solving the steady-state optimization problem in (\ref{eqn:ss_opt}) to obtain the optimal steady state $x_s$ together with the corresponding steady-state input $u_s$. The optimal steady-state cost $\ell_e^*$ was then obtained by computing the value of the economic cost in (\ref{eqn:economic_cost}) using the steady-state values. The asymptotic average performance $\ell_{avg}$ on the other hand was obtained by simulating the closed-loop disturbed system for $T$ time steps and finding the average of the overall control objective in (\ref{eqn:control_objective}) using the equation

\begin{equation} \label{average_cost}
     \ell_{avg} = \frac{1}{T} \sum_{n=0}^{T-1} \ell(x(n),u(n))
\end{equation}
where $T=1000$ time steps. $\ell_{avg}$ considers the effects of the target zone violations and is used to assess the performance of the controllers in the analysis.

\subsection{Effect of risk factor $\delta$}

We first investigate the effect of the design parameter $\delta$ on the optimal steady-state economic cost $\ell_e^*$ and the asymptotic average performance $\ell_{avg}$ of the closed-loop system with the proposed economic zone MPC algorithm.  This was achieved by varying the risk factor $\delta$ and determining $\ell_e^*$ in the associated economic zone as well as $\ell_{avg}$. As can be seen from the results in Figure \ref{fig:risk_factor_effects}, the value of both performance measures generally decrease as the risk factor increases until at $\delta=35$ where $\ell_{avg}$ begins to increase. 

\begin{figure}[hbt!]
  \begin{center}
    \includegraphics[width=0.7\textwidth]{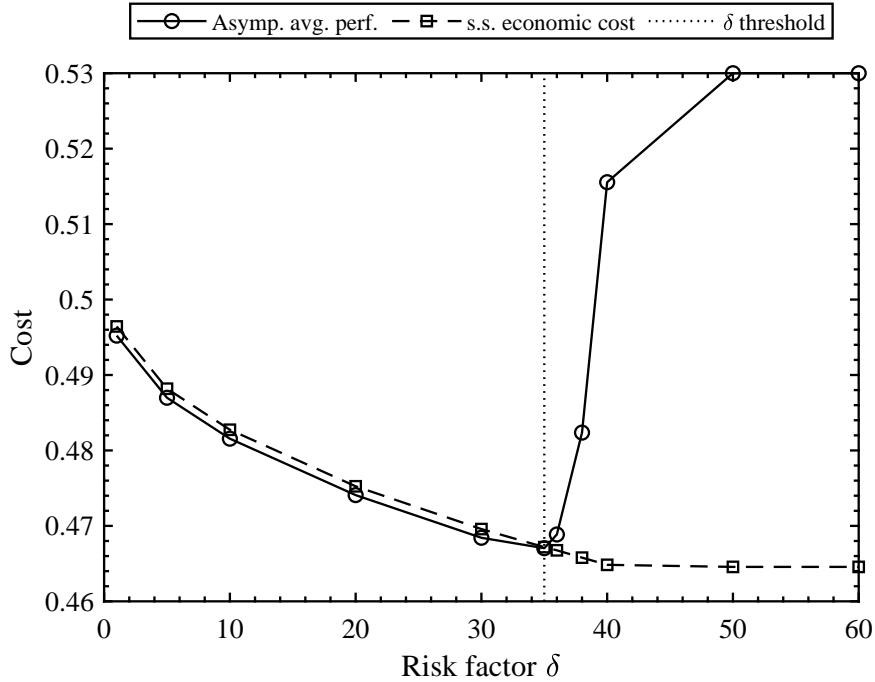}    
    \caption{Effect of risk factor on the best steady-state cost in the economic zone and the closed-loop asymptotic average performance. The dotted lines show the theshold value of the risk factor after which the closed-loop asymptotic average performance begins to deteriorate implying a violation of the target zone. (Solid line with circle markers: Asymptotic average performance, Dashed line with square markers: Optimal steady-state cost, Dotted line: Risk factor theshold)} 
    \label{fig:risk_factor_effects}
  \end{center}
\end{figure}

\begin{figure}[hbt!]
  \begin{center}
    \includegraphics[width=0.7\textwidth]{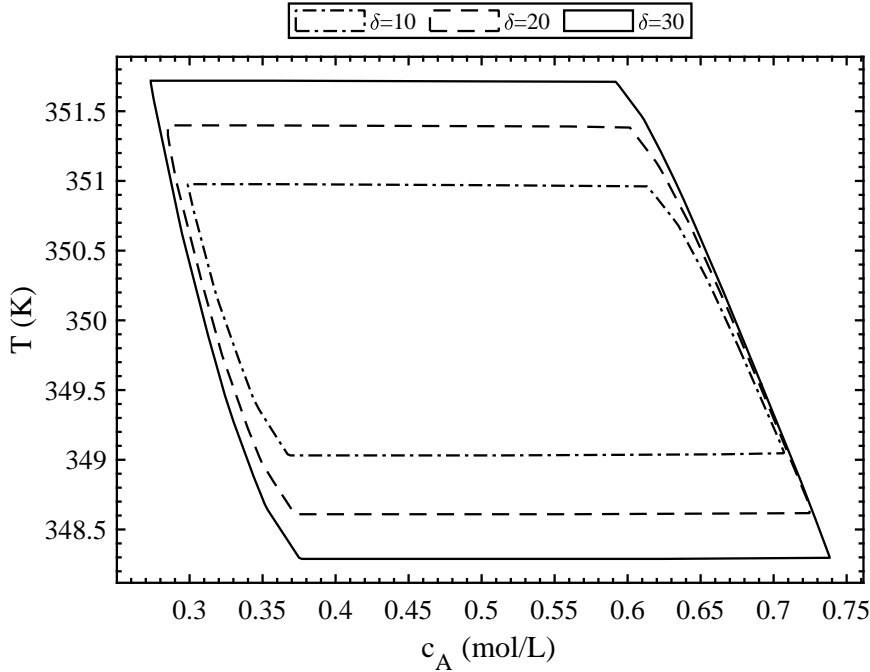}    
    \caption{Effect of risk factor on the economic zone. As the $\delta$ increased, the size of the economic zone also increased and vice versa. The magnitude of the risk factor $\delta$ therefore determines the size of the economic zone and ultimately the conservativeness of the controller. (Solid line: $\delta=30$, Dashed line: $\delta=20$, Dash-dotted line: $\delta=10$)} 
    \label{fig:economic_zones}
  \end{center}
\end{figure}

As mentioned earlier and as shown in Figure \ref{fig:economic_zones}, the size of the economic zone increases as the risk factor increases. This implies that a controller designed with a larger risk factor has a larger operating room to optimize the process economics compared to a controller designed with a smaller risk factor. To explain the reason for the difference in the plots of optimal steady-state economic cost $\ell_e^*$ and the asymptotic average performance $\ell_{avg}$, we look at how the values were obtained. The optimal steady-state economic cost was obtained by solving the static optimzation problem of (\ref{eqn:ss_opt}). Since the effects of the disturbances are not explicitly considered in the steady-state optimization problem, $\ell_e^*$ represents the potentially achievable economic cost. $\ell_{avg}$ on the other hand represents the actual cost achieved in the closed-loop system affected by the disturbance. A large economic zone therefore allowed the process to operate close to the target zone. The presence of the disturbance caused the process to violate the target zone and this resulted in a poor economic performance (on average). A voilation of the target zone means that the conditions in Theorem \ref{thm:2} were not satisfied.

The analysis in Figure \ref{fig:risk_factor_effects} implies that the risk factor should not be arbitrarily chosen. It should be chosen such that the conditions in Theorem \ref{thm:2} are satisfied to ensure that the states of the system converge to the target zone $\m X_t$ in finite time and stays in it thereafter even in the presence of disturbances. For this illustrative example, any $\delta$ value above the threshold value of 35 resulted in a poor closed-loop asymptotic average performance $\ell_{avg}$. Intuitively, the risk factor is a design parameter which offers a trade-off between a conservative controller or a more risk-taking one to maximize the economic objective. This implies that the risk factor $\delta$ in the proposed controller needs to be carefully tuned to get a good trade-off.

\subsection{Comparison with an EMPC tracking the target zone}
Following the analysis of the effects of the risk factor on the controller performance, we compare the closed-loop performance of our proposed controller (ZEMPC tracking the economic zone) to that of conventional EMPC (ZEMPC tracking the target zone). The economic zone $\m X_e$ for our proposed controller was determined using a risk factor of 30. However, as mentioned in Remark \ref{rmk:boundary_steady_state}, the optimal steady state for this process lies on the boundary of $\m X_e$. To ensure that the conditions in the Theorems are satisfied, a smaller economic zone with $\delta = 10$ was computed and the optimal steady state within the smaller economic zone determined. Figure \ref{fig:operating_zones} shows the sets used in the proposed control algorithm with the optimal steady state within its interior. The conventional EMPC on the other hand was designed to track the original target zone without any modification. In both cases, $c_1$ and $c_2$ were selected to be $0$ and $10$ respectively. The steady state values with and without the computed economic zone are $(x_s,u_s) = ([0.483~350.970]^T,299.709)$ and $(x_s,u_s) = ([ 0.465~352.000]^T, 299.413)$ respectively.

\begin{figure}
  \begin{center}
    \includegraphics[width=0.7\textwidth]{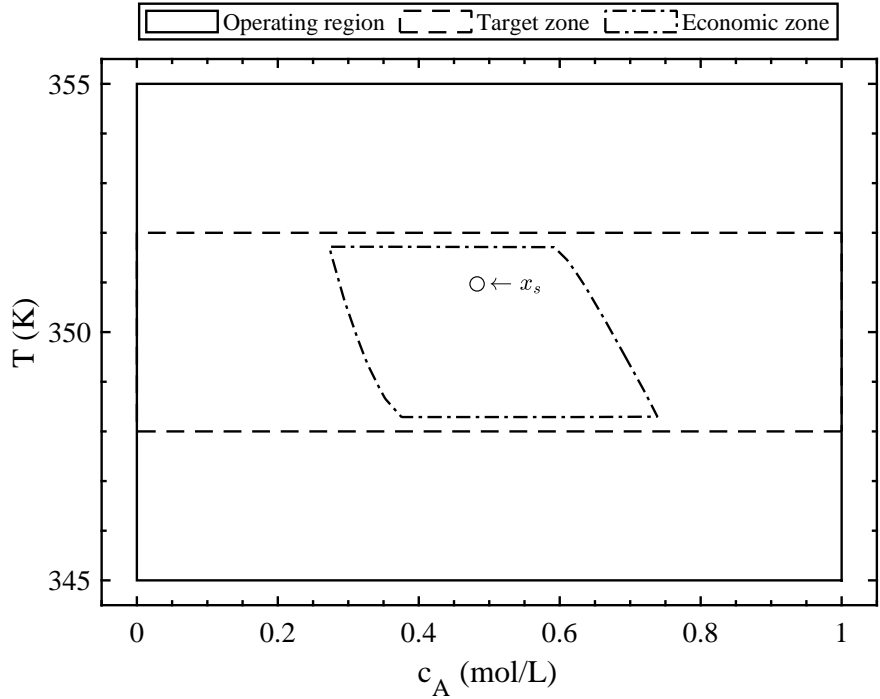}    
    \caption{The sets used in the controller design. (Solid line: Hard constraint, Dashed line: Original zone, Dash-dotted line: Economic zone)} 
    \label{fig:operating_zones}
  \end{center}
\end{figure}

The results of the comparison is shown in Table \ref{tb:average_costs}. As can be seen in the table, our proposed controller gave (on average) a lower asymptotic average performance compared to the conventional EMPC in the presence of the disturbance. 
To understand why this is so, Figure \ref{fig:comparison_trajectories} has been provided. Figure \ref{fig:comparison_trajectories} shows the state, input and average performance trajectories of the closed-loop system under the two controllers in the presence of disturbance. It can be observed that our proposed EMPC forces the system to operate at a temperature below the $352.0K$ thus allowing room for the disturbances to occur without any significant effects on the cost. This results in a fairly stable process economics. The conventional EMPC on the other hand operated close to $352.0K$. Thus, the effects of the disturbances caused the system to operate in an expensive zone which results in a much higher cost. 

\begin{table}
  \begin{center}
  \caption{Asymptotic average performance for the controllers}\label{tb:average_costs}
  \vspace{2mm}
    \begin{tabular}{ccc}
        \hline
    Controller & $\ell_{avg}$ \\  
        \hline
    Conventional EMPC & $0.530$ \\ 
        Proposed EMPC & $0.482$ \\ 
        \hline
    \end{tabular}
  \end{center}
\end{table}

\begin{figure}
  \begin{center}
    \includegraphics[width=0.8\textwidth]{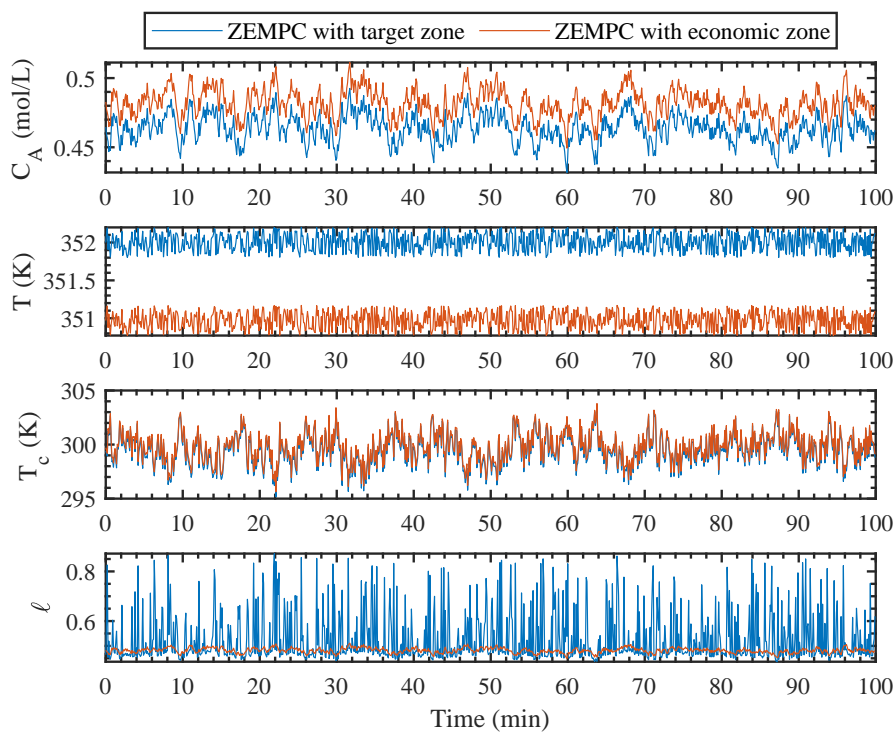}    
    \caption{State, input and economic cost trajectories of the CSTR process under conventional zone EMPC (blue) and our proposed zone EMPC (red)} 
    \label{fig:comparison_trajectories}
  \end{center}
\end{figure}

\section{Concluding remarks}
In this work, we presented a robust EMPC framework with zone tracking for general nonlinear systems. The proposed design ensures that the zone tracking objective can be achieved in finite steps and the economic performance in the operation is optimized. A robust control invariant set within the original target zone is determined and is used as the actual zone tracked. To optimize the economic performance within the zone in the presence of disturbances, the notion of risk factor in the controller design was adopted. An algorithm to determine the economic zone to be tracked was provided. The risk factor determines the conservativeness of the controller and provides a way to tune the EMPC for better economic performance. A nonlinear chemical example was presented to demonstrate the performance of the proposed formulation.

In the future work, it would be of interest to explore an extension of the proposed approach to processes whose optimal operation is not necessarily steady-state but time-varying. Finally, since the proposed controller depends on finding robust control invariant sets (which are generally difficult to obtain for large-scale systems), it will be worth exploring simpler and cheaper ways of obtaining the economic zone. 

\section{Acknowledgement}
This work is supported in part by the Natural Sciences and Engineering Research Council of Canada.

\bibliographystyle{ieeetr}
\bibliography{ref}

\end{document}